\documentclass[psamsfonts]{amsart}

\usepackage{graphicx} 
\usepackage{a4wide}
\usepackage{csquotes}
\usepackage{placeins}
\usepackage{floatrow}
\usepackage{subcaption} 

\usepackage{cite}

\usepackage{amsfonts,amssymb,amscd,amsmath,enumerate,verbatim,}
\usepackage{amssymb}

\usepackage{tikz}
\usepackage{algorithm}
\usepackage[noend]{algpseudocode}
\usetikzlibrary{matrix,arrows,decorations.pathmorphing, chains, positioning,scopes}
\usetikzlibrary{shadows}
\usetikzlibrary{matrix}
\usepackage{graphicx}

\theoremstyle{definition}
\newtheorem{defn}{Definition}[section]
\newtheorem{definition}[defn]{Definition}

\newtheorem{example}[defn]{Example}

\theoremstyle{plain}

\newtheorem{proposition}[defn]{Proposition}

\newtheorem*{theorem*}{Theorem}

\newtheorem{corollary}[defn]{Corollary}
\newtheorem{summary}[defn]{Summary}

\theoremstyle{remark}

\usepackage{parskip}
\algnewcommand{\LineComment}[1]{\State \(\triangleright\) #1}
\tikzstyle{vertex}=[circle, draw]
\usepackage{tikz}
\usepackage{float}

\begin{document}

\title[Ordered DAGS: HyperCubeSort]{Ordered DAGS: HyperCubeSort}

\author[Mikhail Gudim]{Mikhail Gudim} 
\address{} \email{mgudim@gmail.com}


\date{\today} 

\maketitle

\begin{abstract}
We generalize the insertion into a binary heap to any directed acyclic graph (DAG) with one source vertex. This lets us formulate a general method for converting any such DAG into a data structure with priority queue interface. We apply our method to a hypercube DAG to obtain a sorting algorithm of complexity $\mathcal{O}(n\log^2(n))$. As another curious application, we derive a relationship between length of longest path and maximum degree of a vertex in a DAG.
\end{abstract}

\section{Introduction}
Consider a sorted linked list, a binary heap and a Young tableau (see Problem $6-3$ in ~\cite{clrs}). The process of inserting a new element is very similar for all three: we repeatedly exchange newly added element with one of its neighbours until it is in correct place. This simple observation is the main motivation for the present work.

Now let us briefly outline the contents. In Section \ref{sec:term} we fix notation and terminology for the entire paper, in particular we define the notion of an ordered DAG. Our main technical result is in Section \ref{sec:maintaining} where we prove that the structure of an ordered DAG can be easily maintained. This allows us to construct a data structure with priority queue interface from any ordered DAG (Section \ref{sec:data_structure}). Section \ref{sec:dag_sort} demonstrates that some classical algorithms can be viewed as a special case of our general construction. The interaction between sorting and DAGs can be applied to prove statement about DAGs. As an example, we derive a relationship between the maximum degree of a vertex in a DAG and length of longest path (Corollary \ref{cor:dags}). The most juicy part of the paper is Section \ref{sec:hypercube}. There we apply our method to a case where underlying DAG is a hypercube and arrive at a sorting algorithm \textproc{HypercubeSort}, which to our knowledge has not yet been described. In Proposition \ref{prop:comp} we derive an exact expression for the complexity of \textproc{HypercubeSort} in the worst case. Asymptotically it is $\mathcal{O}(n\log^2(n))$ . 

The Java implementation of \textproc{HypercubeSort} is available at ~\cite{hypercube_java}.

\section{Terminology and Notation}
\label{sec:term}

\begin{definition}
    Let $G$ be a DAG and suppose that each node $v$ of $G$ has an integer attribute $v.label$. We call such a DAG \textbf{labeled DAG}. We denote the multiset of all labels in a labeled DAG $G$ by $labels(G)$. Let $(u, v)$ be an edge in a labeled DAG $G$. We use the following terminology:
    \begin{enumerate}
        \item Vertex $u$ is a \textbf{previous neighbour} of $v$ and $v$ is a \textbf{next neighbour} of $u$.
        \item If $u.label \leq v.label$ we say that $(u, v)$ is a \textbf{good} edge. Otherwise, $(u, v)$ is a \textbf{bad} edge. If $(u, v)$ is a bad edge, we call $u$ \textbf{violating previous neighbour} of $v$ and we call $v$ \textbf{violating next neighbour} of $u$.
        \item Labeled DAG $G$ is called \textbf{ordered} if all its edges are good.
    \end{enumerate}
\end{definition}

\begin{example}
\label{ex:labeled_dag}
An example of a ordered DAG is shown in Figure \ref{fig:ordered_dag}. Note that we allow equal labels for vertices.
    \begin{figure}[H]
    \centering
        \begin{tikzpicture}[]
            \node[vertex][](v1) at (0, 1) {$1$};
            \node[vertex][](v2) at (0, -1) {$2$};
            \node[vertex][](v3) at (1, 0) {$4$};
            \node[vertex][](v4) at (2, 1) {$6$};
            \node[vertex][](v5) at (2, -1) {$6$};
            \node[vertex][](v6) at (3, 0) {$8$};
            \node[vertex][](v7) at (4, 1) {$8$};
            \node[vertex][](v8) at (4, 0) {$10$};
            \node[vertex][](v9) at (4, -1) {$9$};
            \node[vertex][](v10) at (5, 0) {$12$};
            \node[vertex][](v11) at (6, 1) {$14$};
            \node[vertex][](v12) at (6, -1) {$16$};
    
            \draw[->, >=stealth] (v1) edge (v3);
            \draw[->, >=stealth] (v2) edge (v3);
            \draw[->, >=stealth] (v2) edge[bend right] (v9);
            \draw[->, >=stealth] (v3) edge (v4);
            \draw[->, >=stealth] (v3) edge (v5);
            \draw[->, >=stealth] (v4) edge (v6);
            \draw[->, >=stealth] (v4) edge[bend left] (v7);
            \draw[->, >=stealth] (v5) edge (v6);
            \draw[->, >=stealth] (v6) edge (v7);
            \draw[->, >=stealth] (v6) edge (v8);
            \draw[->, >=stealth] (v6) edge (v9);
            \draw[->, >=stealth] (v7) edge (v8);
            \draw[->, >=stealth] (v7) edge (v10);
            \draw[->, >=stealth] (v8) edge (v10);
            \draw[->, >=stealth] (v9) edge (v8);
            \draw[->, >=stealth] (v9) edge (v10);
            \draw[->, >=stealth] (v10) edge (v11);
            \draw[->, >=stealth] (v10) edge (v12);
        \end{tikzpicture}
    \caption{An example of ordered dag.}
    \label{fig:ordered_dag}
    \end{figure}
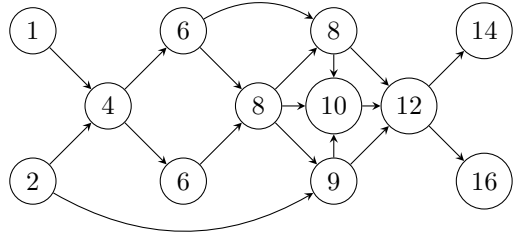
\end{example}

\section{Maintaining Ordered DAGs}
\label{sec:maintaining}
We now describe the procedure \textproc{lowerLabel} which is the workhorse of the entire paper. Generally speaking, it is just a generalization of insertion into a heap, but nevertheless we take care to prove its correctness rigorously. In the following pseudocode we assume that we have a procedure \textproc{getLargestViolating($G$, $v$)} which when given a pointer $v$ to a vertex in a labeled DAG $G$ returns the pointer to a vertex $u$ which is a violating previous neighbour of $v$ with largest $label$ attribute. If there is no previous violating neighbours (this includes the case when there are no previous neighbours at all), the procedure returns $null$.

\begin{algorithm}[H]
    \caption{lowerLabel}\label{alg:lowerLabel}
    \begin{algorithmic}[1]
        \Procedure{lowerLabel}{$G, v, newLabel$}
            \LineComment{
                $G$ is a labeled DAG and $v$ is a pointer to a vertex in $G$ and $newLabel$ is an integer less than $v.label$
            }
            \State $v.label = newLabel$
            \State $current = v$
            \State $violating = true$
            \While{$violating$}
                \State $
                    largestViolating = \textrm{
                        \textproc{getLargestViolating($G$, $current$)}
                    }
                $
                \If{$largestViolating == null$}
                    \State $violating = false$
                \Else
                    \State exchange $current$ with $largestViolating$
                \EndIf
            \EndWhile
            \Return
        \EndProcedure
    \end{algorithmic}
\end{algorithm}

\begin{proposition}
    Let $G$ be an ordered DAG, $v$ a vertex in $G$ and $newLabel$ an integer less than $v.label$.
    \begin{enumerate}
    \item The procedure \textproc{lowerLabel($G$, $v$, $newLabel$)} terminates.
    \item After the termination $G$ remains an ordered DAG.
    \item Let $L$ denote the multiset of labels of $G$ before the call to \textproc{lowerLabel} and $L'$ denote the multiset of labels of $G$ after the call. Then $L'$ is $L$ with $v.label$ replaced by $newLabel$.
    \end{enumerate}
\end{proposition}

\begin{example}
    Before the proof, it would be illustrative to look at an example. Figure \ref{fig:lowerLabel} shows the execution of \textproc{lowerLabel} on a DAG from Example \ref{ex:labeled_dag}.
    
    \begin{figure}[h]
    \centering
        \begin{subfigure}[t]{0.45\textwidth}
            \begin{tikzpicture}
                \node[vertex][](v1) at (0, 1) {$1$};
                \node[vertex][](v2) at (0, -1) {$2$};
                \node[vertex][](v3) at (1, 0) {$4$};
                \node[vertex][](v4) at (2, 1) {$6$};
                \node[vertex][](v5) at (2, -1) {$6$};
                \node[vertex][](v6) at (3, 0) {$8$};
                \node[vertex][](v7) at (4, 1) {$8$};
                \node[vertex][fill=black, text=white](v8) at (4, 0) {$10$};
                \node[vertex][](v9) at (4, -1) {$9$};
                \node[vertex][fill=gray!30](v10) at (5, 0) {$3$};
                \node[vertex][](v11) at (6, 1) {$14$};
                \node[vertex][](v12) at (6, -1) {$16$};
    
                \draw[->, >=stealth] (v1) edge (v3);
                \draw[->, >=stealth] (v2) edge (v3);
                \draw[->, >=stealth] (v2) edge[bend right] (v9);
                \draw[->, >=stealth] (v3) edge (v4);
                \draw[->, >=stealth] (v3) edge (v5);
                \draw[->, >=stealth] (v4) edge (v6);
                \draw[->, >=stealth] (v4) edge[bend left] (v7);
                \draw[->, >=stealth] (v5) edge (v6);
                \draw[->, >=stealth] (v6) edge (v7);
                \draw[->, >=stealth] (v6) edge (v8);
                \draw[->, >=stealth] (v6) edge (v9);
                \draw[->, >=stealth] (v7) edge (v8);
                \draw[->, >=stealth] (v7) edge (v10);
                \draw[->, >=stealth] (v8) edge (v10);
                \draw[->, >=stealth] (v9) edge (v8);
                \draw[->, >=stealth] (v9) edge (v10);
                \draw[->, >=stealth] (v10) edge (v11);
                \draw[->, >=stealth] (v10) edge (v12);
            \end{tikzpicture}
            \caption{After execution of line 4}
        \end{subfigure}
        \hspace{\fill}
        \begin{subfigure}[t]{0.45\textwidth}
            \begin{tikzpicture}
                \node[vertex][](v1) at (0, 1) {$1$};
                \node[vertex][](v2) at (0, -1) {$2$};
                \node[vertex][](v3) at (1, 0) {$4$};
                \node[vertex][](v4) at (2, 1) {$6$};
                \node[vertex][](v5) at (2, -1) {$6$};
                \node[vertex][](v6) at (3, 0) {$8$};
                \node[vertex][](v7) at (4, 1) {$8$};
                \node[vertex][fill=gray!30](v8) at (4, 0) {$3$};
                \node[vertex][fill=black, text=white](v9) at (4, -1) {$9$};
                \node[vertex][](v10) at (5, 0) {$10$};
                \node[vertex][](v11) at (6, 1) {$14$};
                \node[vertex][](v12) at (6, -1) {$16$};
    
                \draw[->, >=stealth] (v1) edge (v3);
                \draw[->, >=stealth] (v2) edge (v3);
                \draw[->, >=stealth] (v2) edge[bend right] (v9);
                \draw[->, >=stealth] (v3) edge (v4);
                \draw[->, >=stealth] (v3) edge (v5);
                \draw[->, >=stealth] (v4) edge (v6);
                \draw[->, >=stealth] (v4) edge[bend left] (v7);
                \draw[->, >=stealth] (v5) edge (v6);
                \draw[->, >=stealth] (v6) edge (v7);
                \draw[->, >=stealth] (v6) edge (v8);
                \draw[->, >=stealth] (v6) edge (v9);
                \draw[->, >=stealth] (v7) edge (v8);
                \draw[->, >=stealth] (v7) edge (v10);
                \draw[->, >=stealth] (v8) edge (v10);
                \draw[->, >=stealth] (v9) edge (v8);
                \draw[->, >=stealth] (v9) edge (v10);
                \draw[->, >=stealth] (v10) edge (v11);
                \draw[->, >=stealth] (v10) edge (v12);
            \end{tikzpicture}
            \caption{End of iteration 1}
        \end{subfigure}
    
        \begin{subfigure}[t]{0.45\textwidth}
            \begin{tikzpicture}
                \node[vertex][](v1) at (0, 1) {$1$};
                \node[vertex][](v2) at (0, -1) {$2$};
                \node[vertex][](v3) at (1, 0) {$4$};
                \node[vertex][](v4) at (2, 1) {$6$};
                \node[vertex][](v5) at (2, -1) {$6$};
                \node[vertex][fill=black, text=white](v6) at (3, 0) {$8$};
                \node[vertex][](v7) at (4, 1) {$8$};
                \node[vertex][](v8) at (4, 0) {$9$};
                \node[vertex][fill=gray!30](v9) at (4, -1) {$3$};
                \node[vertex][](v10) at (5, 0) {$10$};
                \node[vertex][](v11) at (6, 1) {$14$};
                \node[vertex][](v12) at (6, -1) {$16$};
    
                \draw[->, >=stealth] (v1) edge (v3);
                \draw[->, >=stealth] (v2) edge (v3);
                \draw[->, >=stealth] (v2) edge[bend right] (v9);
                \draw[->, >=stealth] (v3) edge (v4);
                \draw[->, >=stealth] (v3) edge (v5);
                \draw[->, >=stealth] (v4) edge (v6);
                \draw[->, >=stealth] (v4) edge[bend left] (v7);
                \draw[->, >=stealth] (v5) edge (v6);
                \draw[->, >=stealth] (v6) edge (v7);
                \draw[->, >=stealth] (v6) edge (v8);
                \draw[->, >=stealth] (v6) edge (v9);
                \draw[->, >=stealth] (v7) edge (v8);
                \draw[->, >=stealth] (v7) edge (v10);
                \draw[->, >=stealth] (v8) edge (v10);
                \draw[->, >=stealth] (v9) edge (v8);
                \draw[->, >=stealth] (v9) edge (v10);
                \draw[->, >=stealth] (v10) edge (v11);
                \draw[->, >=stealth] (v10) edge (v12);
            \end{tikzpicture}
            \caption{End of iteration 2}
        \end{subfigure}
        \hspace{\fill}
        \begin{subfigure}[t]{0.45\textwidth}
            \begin{tikzpicture}
                \node[vertex][](v1) at (0, 1) {$1$};
                \node[vertex][](v2) at (0, -1) {$2$};
                \node[vertex][](v3) at (1, 0) {$4$};
                \node[vertex][fill=black, text=white](v4) at (2, 1) {$6$};
                \node[vertex][](v5) at (2, -1) {$6$};
                \node[vertex][fill=gray!30](v6) at (3, 0) {$3$};
                \node[vertex][](v7) at (4, 1) {$8$};
                \node[vertex][](v8) at (4, 0) {$9$};
                \node[vertex][](v9) at (4, -1) {$8$};
                \node[vertex][](v10) at (5, 0) {$10$};
                \node[vertex][](v11) at (6, 1) {$14$};
                \node[vertex][](v12) at (6, -1) {$16$};
    
                \draw[->, >=stealth] (v1) edge (v3);
                \draw[->, >=stealth] (v2) edge (v3);
                \draw[->, >=stealth] (v2) edge[bend right] (v9);
                \draw[->, >=stealth] (v3) edge (v4);
                \draw[->, >=stealth] (v3) edge (v5);
                \draw[->, >=stealth] (v4) edge (v6);
                \draw[->, >=stealth] (v4) edge[bend left] (v7);
                \draw[->, >=stealth] (v5) edge (v6);
                \draw[->, >=stealth] (v6) edge (v7);
                \draw[->, >=stealth] (v6) edge (v8);
                \draw[->, >=stealth] (v6) edge (v9);
                \draw[->, >=stealth] (v7) edge (v8);
                \draw[->, >=stealth] (v7) edge (v10);
                \draw[->, >=stealth] (v8) edge (v10);
                \draw[->, >=stealth] (v9) edge (v8);
                \draw[->, >=stealth] (v9) edge (v10);
                \draw[->, >=stealth] (v10) edge (v11);
                \draw[->, >=stealth] (v10) edge (v12);
            \end{tikzpicture}
            \caption{End of iteration 3}
        \end{subfigure}
    
        \begin{subfigure}[t]{0.45\textwidth}
            \begin{tikzpicture}
                \node[vertex][](v1) at (0, 1) {$1$};
                \node[vertex][](v2) at (0, -1) {$2$};
                \node[vertex][fill=black, text=white](v3) at (1, 0) {$4$};
                \node[vertex][fill=gray!30](v4) at (2, 1) {$3$};
                \node[vertex][](v5) at (2, -1) {$6$};
                \node[vertex][](v6) at (3, 0) {$6$};
                \node[vertex][](v7) at (4, 1) {$8$};
                \node[vertex][](v8) at (4, 0) {$9$};
                \node[vertex][](v9) at (4, -1) {$8$};
                \node[vertex][](v10) at (5, 0) {$10$};
                \node[vertex][](v11) at (6, 1) {$14$};
                \node[vertex][](v12) at (6, -1) {$16$};
    
                \draw[->, >=stealth] (v1) edge (v3);
                \draw[->, >=stealth] (v2) edge (v3);
                \draw[->, >=stealth] (v2) edge[bend right] (v9);
                \draw[->, >=stealth] (v3) edge (v4);
                \draw[->, >=stealth] (v3) edge (v5);
                \draw[->, >=stealth] (v4) edge (v6);
                \draw[->, >=stealth] (v4) edge[bend left] (v7);
                \draw[->, >=stealth] (v5) edge (v6);
                \draw[->, >=stealth] (v6) edge (v7);
                \draw[->, >=stealth] (v6) edge (v8);
                \draw[->, >=stealth] (v6) edge (v9);
                \draw[->, >=stealth] (v7) edge (v8);
                \draw[->, >=stealth] (v7) edge (v10);
                \draw[->, >=stealth] (v8) edge (v10);
                \draw[->, >=stealth] (v9) edge (v8);
                \draw[->, >=stealth] (v9) edge (v10);
                \draw[->, >=stealth] (v10) edge (v11);
                \draw[->, >=stealth] (v10) edge (v12);
            \end{tikzpicture}
            \caption{End of iteration 4}
        \end{subfigure}
        \hspace{\fill}
        \begin{subfigure}[t]{0.45\textwidth}
            \begin{tikzpicture}
                \node[vertex][](v1) at (0, 1) {$1$};
                \node[vertex][](v2) at (0, -1) {$2$};
                \node[vertex][fill=gray!30](v3) at (1, 0) {$3$};
                \node[vertex][](v4) at (2, 1) {$4$};
                \node[vertex][](v5) at (2, -1) {$6$};
                \node[vertex][](v6) at (3, 0) {$6$};
                \node[vertex][](v7) at (4, 1) {$8$};
                \node[vertex][](v8) at (4, 0) {$9$};
                \node[vertex][](v9) at (4, -1) {$8$};
                \node[vertex][](v10) at (5, 0) {$10$};
                \node[vertex][](v11) at (6, 1) {$14$};
                \node[vertex][](v12) at (6, -1) {$16$};
    
                \draw[->, >=stealth] (v1) edge (v3);
                \draw[->, >=stealth] (v2) edge (v3);
                \draw[->, >=stealth] (v2) edge[bend right] (v9);
                \draw[->, >=stealth] (v3) edge (v4);
                \draw[->, >=stealth] (v3) edge (v5);
                \draw[->, >=stealth] (v4) edge (v6);
                \draw[->, >=stealth] (v4) edge[bend left] (v7);
                \draw[->, >=stealth] (v5) edge (v6);
                \draw[->, >=stealth] (v6) edge (v7);
                \draw[->, >=stealth] (v6) edge (v8);
                \draw[->, >=stealth] (v6) edge (v9);
                \draw[->, >=stealth] (v7) edge (v8);
                \draw[->, >=stealth] (v7) edge (v10);
                \draw[->, >=stealth] (v8) edge (v10);
                \draw[->, >=stealth] (v9) edge (v8);
                \draw[->, >=stealth] (v9) edge (v10);
                \draw[->, >=stealth] (v10) edge (v11);
                \draw[->, >=stealth] (v10) edge (v12);
            \end{tikzpicture}
            \caption{End of iteration 5. The DAG $G$ is now ordered, procedure exits.}
        \end{subfigure}
        \caption{
            Procedure \textproc{lowerLabel} applied to ordered DAG from Figure \ref{fig:ordered_dag} to lower label of a vertex from $12$ to $3$.
            The vertex $current$ is highlighted in gray and the vertex that will be returned by \textproc{getLargestViolating} at next iteration is black.
        \label{fig:lowerLabel}
    }
    \end{figure}
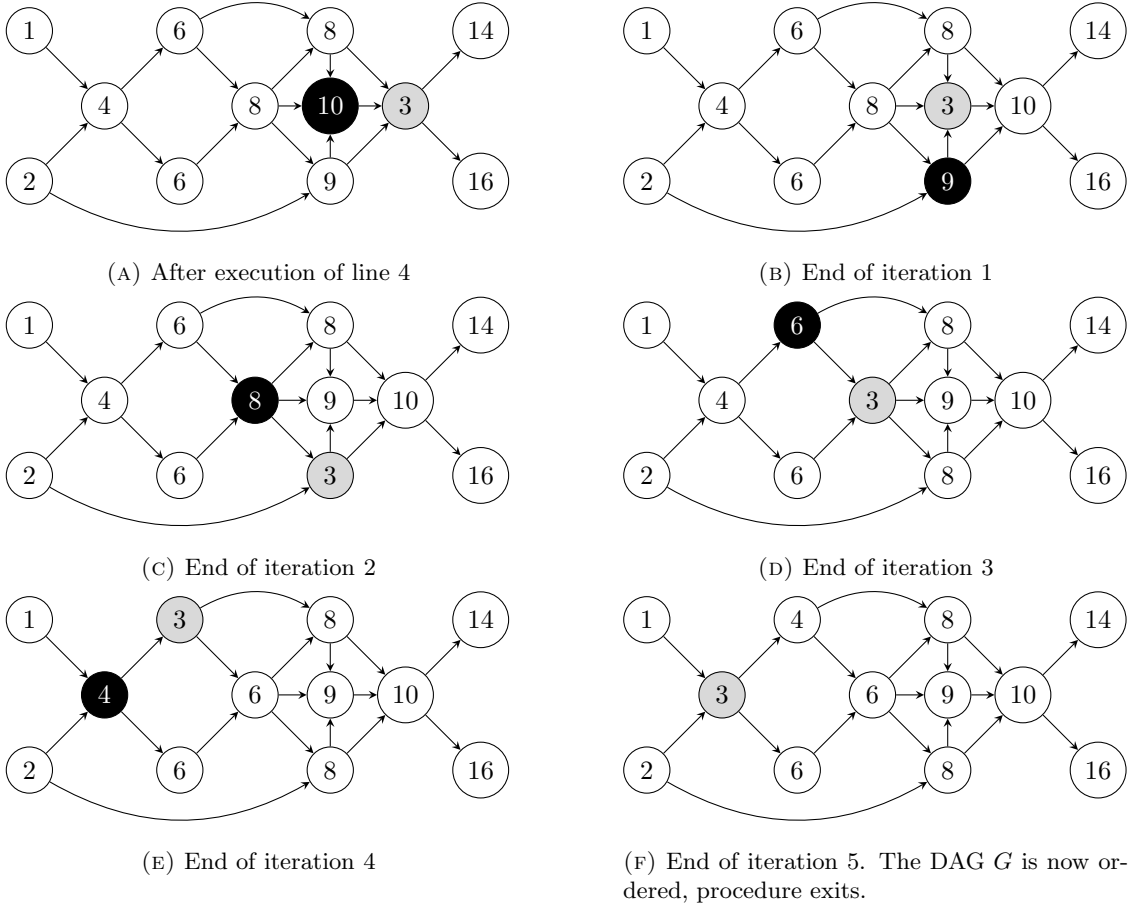
\end{example}

\begin{proof}
    It is easy to see that the procedure terminates because $G$ is a finite DAG. Statement (3) is also clear, since the only step which changes the multiset of labels is in line 3 of the pseudocode.

    We prove (2) by showing that the \textproc{lowerLabel} maintains the following two-part invariant:
    \begin{displayquote}
        At the end of each iteration of the \textbf{while} loop in lines 6-11:
        \begin{enumerate}
        \item All bad edges in $G$ (if any) are entering $current$.
        \item For any previous neighbour $p$ of $current$ and any next neighbour $n$ of $current$ $p.label \leq n.label$.
        \end{enumerate}
    \end{displayquote}

    Before line 4 $G$ is ordered. Right after line 4 executes $current$ is $v$ and it is the only vertex whose $label$ attribute changed, so part (2) of the invariant is maintained and the only edges that could become bad are those entering and leaving $current$. Since the $label$ of $v$ becomes smaller, all edges leaving $current$ remain good, but edges entering $current$ could become bad, so part (1) of the invariant is maintained.

    Now assume the invariant was maintained for the first $m$ iterations of the \textbf{while} loop. We show that it is maintained after $(m + 1)$-st iteration. If \textproc{getLargestViolating} returns $null$ it means that there are no bad edges entering $current$ and by part (1) of the invariant $G$ is ordered DAG. The procedure exists. Now we consider the case when \textproc{getLargestViolating} returns non-null value.
    
    Let $x$ denote the vertex in $current$ variable before the exchange in line 11. Let $\alpha = x.label$, denote by $p_1, p_2, \dots, p_k$ previous neighbours of $x$, by $n_1, n_2, \dots, n_l$ next neighbours of $x$ and suppose \textproc{getLargestViolating} returns $p_k$ with $p_k.label = \beta$. By definition of \textproc{getLargestViolating} the following inequalities are true:

    \begin{equation}
        \begin{gathered}
            \beta > \alpha\\
            \textrm{ (because $p_k$ is violating)}
        \end{gathered}
        \label{ineq:violating}
    \end{equation}
    and
    \begin{equation}
        \begin{gathered}
            \beta \geq p_i.label
            \textrm{ for all $i$ with $1 \leq i \leq (k - 1)$}\\
            \textrm{ (because $p.k$ is largest violating previous neighbour)}
            \label{ineq:largest_violating}
        \end{gathered}
    \end{equation}

    Because part (1) of the invariant was maintained up to this moment, the following inequalities are true:

    \begin{equation}
        \alpha \leq n_j.label \textrm{ for all $j$ with $1 \leq j \leq l$}
        \label{ineq:inv_part1}
    \end{equation}

    Because part (2) of the invariant was maintained up to this moment, the following inequalities are also true:

    \begin{equation}
        p_i.label \leq n_j.label \textrm{ for all $i$ and $j$}
        \label{ineq:inv_part2}
    \end{equation}

    After the exchange in line 11 $current$ is $p_k$ and $label$ attributes of only two vertices changed, namely $x.label = \beta$ and $p_k.label = \alpha$. All edges entering $x$ are of the form $(p_i, x)$ for $1 \leq i \leq k$. By inequalities in (\ref{ineq:largest_violating}) all edges $(p_i, x)$ for $1 \leq i \leq (k - 1)$ are good and by inequality (\ref{ineq:violating}) the edge $(p_k, x)$ is good. So all the edges entering $x$ are good. All the edges leaving $x$ are of the form $(x, n_j)$ for $1 \leq j \leq l$. They remain good by inequalities in (\ref{ineq:inv_part2}) applied with $i = k$. Now let us look at $current = p_k$. Before the exchange all edges entering and leaving $p_k$ were good. After the exchange, by the inequality (\ref{ineq:violating}) value of $p_k.label$ lowered. This means all the edges leaving $p_k$ remain good and only edges entering $p_k$ could become bad. So part (2) of the invariant is maintained.
\end{proof}

The following pseudocode shows that in abstract sense the processes of lowering and raising value of $label$ attribute are equivalent:

\begin{algorithm}[H]
    \caption{raiseLabel}\label{alg:raiseLabel}
    \begin{algorithmic}[1]
        \Procedure{raiseLabel}{$G, v, newLabel$}
            \LineComment{
                $G$ is a labeled DAG and $v$ is a pointer to a vertex in $G$ and $newLabel$ is an integer greater than $v.label$
            }
            \State Multiply the label of each vertex of $G$ by $(-1)$
            \State Reverse each edge in $G$ \Comment{after this $G$ is ordered}
            \State Run the procedure \textproc{lowerLabel$(G, v, -newLabel)$}
            \State Multiply the label of each vertex of $G$ by $(-1)$
            \State Reverse each edge in $G$
        \EndProcedure
    \end{algorithmic}
\end{algorithm}
In practice one can implement \textproc{raiseLabel} in a completely symmetrical way to \textproc{lowerLabel} by reversing the directions of edges and meaning of comparisons.

\section{Ordered DAG Data Structure}
\label{sec:data_structure}
With Algorithms \ref{alg:lowerLabel} and \ref{alg:raiseLabel} at hand, it is probably clear how to make any DAG $G$ with only one source vertex induces a data structure $\textproc{OrderedDag}_G$. However, there are some details which we do not want to neglect.

The constructor of $\textproc{OrderedDag}_G$ assigns value of $\infty$ to every vertex in the DAG.

To insert a new element with label $l$ in principle one can pick \emph{any} vertex $v$ with $v.label = \infty$ and then call $\textproc{lowerLabel}(v, l)$ to restore the order in $G$. But arbitrary choice is ambiguous and may be non-optimal. Therefore, we also assume that we have a (stateful) procedure $\textproc{getNext}()$ which returns the next vertex of $G$ in breadth-first order. In particular, the first call to \textproc{getNext} returns the source vertex $s$. The second call returns one of next neighbours of $s$. After all neighbours of $s$ were returned, it returns neighbours of neighbours of $s$ and so on. Thus we have:
\begin{algorithm}[H]
    \caption{insert}\label{alg:insert}
    \begin{algorithmic}[1]
        \Procedure{insert}{$l$}
            \LineComment{
                The element $l$ is the new label to be inserted into $G$.
            }
            \State $nextVertex = \textproc{getNext}()$
            \State $\textproc{lowerLabel}(G, nextVertex, l)$
        \EndProcedure
    \end{algorithmic}
\end{algorithm}

Since $G$ is ordered DAG, the source vertex $s$ must have the minimum label. The procedure \textproc{getMin} returns this vertex.

Given a vertex $v$ in $G$ we can remove $v$ by calling $\textproc{raiseLabel}(G, v, \infty))$. In particular, we can remove the minimum value this way:
\begin{algorithm}[H]
    \caption{removeMin}\label{alg:remove_min}
    \begin{algorithmic}[1]
        \Procedure{removeMin}{}
            \LineComment{
                Removes the vertex with minimum $label$ attribute from $G$.
            }
            \State $s = \textproc{getMin}()$
            \State $\textproc{raiseLabel}(G, s, \infty)$
        \EndProcedure
    \end{algorithmic}
\end{algorithm}

For the reference we make the following Summary:
\begin{summary}
    Any DAG $G$ with $N$ vertices and only one source vertex induces a data structure $\textproc{OrderedDag}_G$, which implements the following interface:
    \begin{enumerate}
        \item $\textproc{OrderedDag(G)}$: creates the data structure $\textproc{OrderedDag}_G$, with $v.label = \infty$ for each vertex $v$ in $G$.
        \item $\textproc{insert}(l)$: replaces one of $\infty$'s (if any left) with $l$ in $labels(G)$.
        \item $\textproc{getMin}()$: returns a vertex whose $label$ attribute is a minimum value of $labels(G)$.
        \item $\textproc{removeMin}()$: removes (a) minimum value from $labels(G)$.
        \item $\textproc{lowerLabel}(v, newLabel)$: replaces $v.label$ in $labels(G)$ with $newLabel < v.label$.
        \item $\textproc{raiseLabel}(v, newLabel)$: replaces $v.label$ in $labels(G)$ with $newLabel > v.label$.
    \end{enumerate}
\end{summary}

\section{General \textproc{DAGSort}}
\label{sec:dag_sort}
Given a DAG $G$ with $n$ vertices we can use $\textproc{OrderedDag}_G$ to sort array of $n$ elements.
\begin{algorithm}[H]
    \caption{DAGSort}\label{alg:dag_sort}
    \begin{algorithmic}[1]
        \Procedure{DAGSort}{$G, A$}
            \LineComment{
                $G$ is a DAG with $n$ vertices and $A$ is array with $n$ elements.
            }
            \State $dag = \textproc{OrderedDag}(G)$
            \For {all elements $a$ in $A$}
                \State $dag.\textproc{insert}(a)$
            \EndFor
            \State $A' = \textrm{new array of size $n$}$
            \For {$i$ from $0$ to $(n - 1)$ inclusively}
                \State $A'[i] = dag.\textproc{removeMin}()$
            \EndFor
            \Return {$A'$}
        \EndProcedure
    \end{algorithmic}
\end{algorithm}

For a general DAG we have the following obvious coarse upper-bound on complexity:
\begin{proposition}
\label{prop:gen_comp}
Let $G$ be any DAG with $n$ vertices and only one source vertex $s$. Let $D_{in}$ and $D_{out}$ denote the highest in- and out- degree of a vertex in $G$ respectively and let $L$ denote the maximum length of a simple path starting at $s$. Then $\textproc{DAGSort}(G, \bullet)$ makes at most $nL(D_{in} + D_{out})$ comparisons.
\end{proposition}
\begin{proof}
    During insertion new element will be exchanged with at most $L$ vertices. To make one exchange we need at most $D_{in}$ comparisons, so the cost of inserting all $n$ elements into $\textproc{OrderedDag}_G$ is bounded above by $nD_{in}L$. The other term $nD_{out}L$ comes from extracting minimum element $n$ times.
\end{proof}

Now we put some well-know algorithms in the context of general \textproc{DAGSort}.
\begin{example}
    Let $G$ be a DAG in Figure \ref{fig:selection_sort}. With this $G$ as an underlying DAG the $\textproc{DAGSort}(G, \bullet)$ is the selection sort algorithm. The vertex $s$ always contains the minimum value.
    \begin{figure}[H]
    \centering
        \begin{tikzpicture}[]
            \node[vertex][](v0) at (0, 3) {$s$};
            \node[vertex][](v1) at (-2, 0) {};
            \node[vertex][](v2) at (-1, 0) {};
            \node[vertex][](v3) at (0, 0) {};
            \node[vertex][](v4) at (1, 0) {};
            \node[vertex][](v5) at (2, 0) {};
            \draw[->, >=stealth] (v0) edge (v1);
            \draw[->, >=stealth] (v0) edge (v2);
            \draw[->, >=stealth] (v0) edge (v3);
            \draw[->, >=stealth] (v0) edge (v4);
            \draw[->, >=stealth] (v0) edge (v5);
        \end{tikzpicture}
    \caption{}
    \label{fig:selection_sort}
    \end{figure}
\end{example}

\begin{example}
    Let $G$ be a DAG in Figure \ref{fig:insertion_sort}. With this $G$ as an underlying DAG the $\textproc{DAGSort}(G, \bullet)$ is the insertion sort algorithm.
    \begin{figure}[H]
    \centering
        \begin{tikzpicture}[]
            \node[vertex][](v0) at (0, 0) {$s$};
            \node[vertex][](v1) at (1, 0) {};
            \node[vertex][](v2) at (2, 0) {};
            \node[vertex][](v3) at (3, 0) {};
            \node[vertex][](v4) at (4, 0) {};
            \node[vertex][](v5) at (5, 0) {};
            \draw[->, >=stealth] (v0) edge (v1);
            \draw[->, >=stealth] (v1) edge (v2);
            \draw[->, >=stealth] (v2) edge (v3);
            \draw[->, >=stealth] (v3) edge (v4);
            \draw[->, >=stealth] (v4) edge (v5);
        \end{tikzpicture}
    \caption{}
    \label{fig:insertion_sort}
    \end{figure}
\end{example}

\begin{example}
    Let $G$ be a DAG in Figure \ref{fig:insertion_sort}. With this $G$ as an underlying DAG the $\textproc{DAGSort}(G, \bullet)$ is the sorting algorithm using Young tableau.
    \begin{figure}[H]
    \centering
        \begin{tikzpicture}[]
            \node[vertex][](v0) at (0, 0) {};
            \node[vertex][](v1) at (0, 1) {};
            \node[vertex][](v2) at (0, 2) {};
            \node[vertex][](v3) at (-1, 0) {};
            \node[vertex][](v4) at (-1, 1) {};
            \node[vertex][](v5) at (-1, 2) {};
            \node[vertex][](v6) at (-2, 0) {$s$};
            \node[vertex][](v7) at (-2, 1) {};
            \node[vertex][](v8) at (-2, 2) {};
            \draw[->, >=stealth] (v0) edge (v1);
            \draw[->, >=stealth] (v1) edge (v2);
            \draw[->, >=stealth] (v3) edge (v4);
            \draw[->, >=stealth] (v4) edge (v5);
            \draw[->, >=stealth] (v6) edge (v7);
            \draw[->, >=stealth] (v7) edge (v8);

            \draw[->, >=stealth] (v6) edge (v3);
            \draw[->, >=stealth] (v3) edge (v0);
            \draw[->, >=stealth] (v7) edge (v4);
            \draw[->, >=stealth] (v4) edge (v1);
            \draw[->, >=stealth] (v8) edge (v5);
            \draw[->, >=stealth] (v5) edge (v2);
        \end{tikzpicture}
    \caption{}
    \label{fig:insertion_sort}
    \end{figure}
    One can generalize two-dimensional Young tableaux to $k$-dimensional tableaux: the underlying DAG is a $k$-dimensional grid. If a grid has $n$ elements then the length of longest path is $k\sqrt[k]{n}$ and each vertex has at most $k$ previous and $k$ next neighbours. Thus by upper bound of \ref{prop:gen_comp} the \textproc{DAGSort} using $k$-dimensional Young tableaux is of complexity $\mathcal{O}(kn^{1 + \frac{1}{k}})$.
\end{example}

If we allow very high in-degree of a vertex in a DAG, we can make the longest path in the DAG small - consider DAG with one source vertex connected to other vertices. At the other extreme, consider a DAG where all the vertices are arranged in a linked list: all the vertices have small degree but the longest path is long. What can we say in a generic case?
\begin{corollary}
\label{cor:dags}
    Let $G$ be any DAG with $n \geq 2$ vertices and only one source vertex $s$. Let $D_{in}$ and $D_{out}$ denote the highest in- and out- degree of a vertex in $G$ respectively and let $L$ denote the maximum length of a simple path starting at $s$. Then the following inequality holds:
    \begin{equation*}
        \frac{1}{n} \log(n!) \leq L(D_{in} + D_{out})
    \end{equation*}
\end{corollary}
In other words, to densely pack (length of longest path is small) $n$ vertices into a DAG one cannot avoid vertices with high in-degrees.
\begin{proof}
    Let $G$ be any DAG with $n \geq 2$ vertices and let $\textproc{DAGSort}_G$ denote the sorting algorithm defined by $G$. Let $I$ denote the input on which $\textproc{DAGSort}_G$ makes at least $\log(n!)$ comparisons to sort $I$. Such input always exists (see Section 8-1 of ~\cite{clrs}). Let us denote the number of comparisons made by $\textproc{DAGSort}_G$ to sort $I$ by $T(I)$. We have the following lower bound on $T(I)$:
    \begin{equation}
        \log(n!) \leq T(I)
    \end{equation}
    On the other hand, by the bound in Proposition \ref{prop:gen_comp} we have:
    \begin{equation}
        T(I) \leq nL(D_{in} + D_{out})
    \end{equation}
    Combining the two inequalities above gives the result.
\end{proof}

\section{\textproc{HypercubeSort}}
\label{sec:hypercube}
Let $S$ be a set with $k$ elements. The set of $2^k$ subsets of $S$ forms a DAG: vertices are subsets there is an arrow from subset $S$ to subset $T$ if $T$ is obtained from $S$ by adding one element. We denote this DAG by $DAG(S)$. An example with $S = \{1, 2, 3\}$ is shown in Figure \ref{fig:subset_dag}.

\def \side{3}
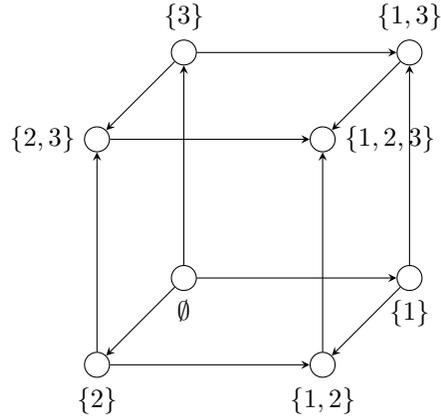
\begin{figure}[H]
\centering
    \begin{tikzpicture}[]
        \node[vertex, label=below:{$\emptyset$}][](v1) at (0, 0, 0) {};
        \node[vertex, label=below:{$\{1\}$}][](v2) at (\side, 0, 0) {};
        \node[vertex, label=above:{$\{3\}$}][](v3) at (0, \side, 0) {};
        \node[vertex, label=below:{$\{2\}$}][](v4) at (0, 0, \side) {};
        \node[vertex, label=above:{$\{1, 3\}$}][](v5) at (\side, \side, 0) {};
        \node[vertex, label=below:{$\{1, 2\}$}][](v6) at (\side, 0, \side) {};
        \node[vertex, label=left:{$\{2, 3\}$}][](v7) at (0, \side, \side) {};
        \node[vertex, label=right:{$\{1, 2, 3\}$}][](v8) at (\side, \side, \side) {};

        \draw[->, >=stealth] (v1) edge (v2);
        \draw[->, >=stealth] (v1) edge (v3);
        \draw[->, >=stealth] (v1) edge (v4);
        \draw[->, >=stealth] (v2) edge (v5);
        \draw[->, >=stealth] (v2) edge (v6);
        \draw[->, >=stealth] (v3) edge (v5);
        \draw[->, >=stealth] (v3) edge (v7);
        \draw[->, >=stealth] (v4) edge (v6);
        \draw[->, >=stealth] (v4) edge (v7);
        \draw[->, >=stealth] (v5) edge (v8);
        \draw[->, >=stealth] (v6) edge (v8);
        \draw[->, >=stealth] (v7) edge (v8);
    \end{tikzpicture}
\caption{$DAG(\{1, 2, 3\})$}
\label{fig:subset_dag}
\end{figure}

The DAG of subsets $DAG(S)$ is a $k$-fold direct product of a DAG in Figure \ref{fig:edge_dag} with itself. Such a DAG is called $k$-dimensional \textbf{hypercube}. Since the vertices can be identified with subsets we have a notion of \textbf{cardinality of a vertex}: we say that a vertex $v$ is hypercube is of cardinality $m$ if $v$ corresponds to an $m$-element subset of $S$. There is only one vertex $s$ of cardinality zero (it corresponds to the empty subset) which is the only source vertex. Let $u$ be a vertex of cardinality $m > 0$. Then length of any path from $s$ to $u$ is $m$, $u$ has $m$ previous neighbours and $(k - m)$ next neighbours.
\begin{figure}[H]
\centering
    \begin{tikzpicture}[]
        \node[vertex][](v1) at (0, 0) {};
        \node[vertex][](v2) at (3, 0) {};
        \draw[->, >=stealth] (v1) edge (v2);
    \end{tikzpicture}
\caption{}
\label{fig:edge_dag}
\end{figure}

We can apply the general \textproc{DAGSort} to the hypercube DAG and we call the resulting algorithm \textproc{HypercubeSort}. To implement \textproc{HypercubeSort} we can identify a vertex of a hypercube with the subset it represents which can be written uniquely as a binary string of length $k$. For example, the subset $\{2, 4, 5\}$ of $S = \{1, 2, 3, 4, 5, 6, 7, 8\}$ is written as $01011000$. In turn, we can convert each binary string to a number, which can serve as an index into array. The procedure \textproc{getNext} can be implemented by first listing all $0$-element subsets, then all $1$-element subsets, then all $2$-element subsets and so on. The exact algorithm that does not require any extra memory and each call takes $\mathcal{O}(1)$ time is described in Section $3$ of ~\cite{comb_alg}. Our Java implementation of \textproc{HypercubeSort} with all the details is available at ~\cite{hypercube_java}, but in our version for simplicity we precompute the order of vertices by a recursive procedure and store the result.

By upper bound in Proposition \ref{prop:gen_comp} it takes $\mathcal{O}(n\log^2(n))$ comparisons for \textproc{HypercubeSort} to sort $n$ element array. Unfortunately, this bound cannot be improved by a more detailed analysis as the proof of the following proposition shows.
\begin{proposition}
    \label{prop:comp}
    To sort $n$ elements the algorithm \textproc{HypercubeSort} makes at most $\mathcal{O}(n\log^2(n))$ comparisons.
\end{proposition}

\begin{proof}
    We assume that $n$ is a power of two. By symmetry between insertion and deletion operations, it is enough to estimate the number of comparisons to insert $n = 2^k$ elements. Denote this number by $T(n)$.
    To insert one vertex of cardinality $i$ it takes in the worst case $i + (i - 1) + \dots +  1 = \frac{1}{2}(i + 1)i$ comparisons. There are exactly $\binom{k}{i}$ vertices of cardinality $i$. Thus
    \begin{equation}
    \label{eq:t_n}
        T(n) = \sum_{i = 0}^{k} \frac{1}{2}\binom{k}{i}(i + 1)i
    \end{equation}

    To evaluate the above summation, consider the function $f$ defined by
    \begin{equation*}
        f(x) = x(1 + x)^k = \sum_{i = 0}^{k} \binom{k}{i}x^{i + 1}
    \end{equation*}
    Take derivatives of the two expressions of $f$:
    \begin{equation*}
        f'(x) = (1 + x)^k + kx(1 + x)^{k - 1} = \sum_{i = 0}^{k} \binom{k}{i}(i + 1)x^i
    \end{equation*}
    And the second derivatives:
    \begin{equation*}
        f''(x) = k(1 + x)^{k - 1} + k(1 + x)^{k - 1} + k(k - 1)x(1 + x)^{k - 2} = \sum_{i = 1}^{k} \binom{k}{i}(i + 1)ix^{i - 1}
    \end{equation*}
    By comparing right-hand side of the last equation with equation (\ref{eq:t_n}) we see that $T(n) = \frac{1}{2}f''(1)$. But we can evaluate $f''(1)$ using formula on the left-hand side of the above equation with $x = 1$. The precise expression we get is
    \begin{equation*}
        k2^k + k(k - 1)2^{k-2}.
    \end{equation*}
    which is clearly $\mathcal{O}(2^kk^2)$.
\end{proof}

\bibliography{/home/mgudim/Documents/ordered_dags/ordered_dags}{}
\bibliographystyle{plain}

\end{document}